\documentclass[table,xcdraw,%
 reprint,
 amsmath,amssymb,
 aps,
]{revtex4-2}

\usepackage{color}
\usepackage{xcolor}
\usepackage{graphicx}
\usepackage{dcolumn}
\usepackage{bm}

\usepackage{float}
\usepackage{graphicx,booktabs,array}
\usepackage{dsfont}
 \usepackage{tikz}
\usepackage{mathtools, amsthm}
\usepackage{pst-func}
\usepackage{subcaption}
\usepackage{graphicx}
\usepackage{hyperref}
\usepackage{anyfontsize}

 \usetikzlibrary{plotmarks}
 \usepackage{pgfplots}
\pgfplotsset{compat=newest}
\usetikzlibrary{patterns}
\usetikzlibrary{decorations.text}
\usepgfplotslibrary{fillbetween}

\newcommand{\Et}{\color{black}}

\newtheorem{proposition}{Proposition}

\newtheorem{corollary}{Corollary}

\newtheorem{remark}{Remark}

\begin{document}

\preprint{APS/123-QED}

\title{The diverse and fair structures of growth with scale-free saturations}

\author{Alain Govaert}
\email{alain.govaert@it.uu.se}
 \affiliation{Department of Information Technology, Uppsala University.}

\author{Emma Tegling}
\email{emma.tegling@control.lth.se}
\affiliation{Department of Automatic Control, Lund University.}

\author{Andr\'{e} Teixeira}
\email{andre.teixeira@it.uu.se}
 \affiliation{Department of Information Technology, Uppsala University.}


\date{\today}

\begin{abstract}
Real-world growth processes and scalings have been broadly categorized into three growth regimes with distinctly different properties and driving forces. The first two are characterized by a positive 
and constant feedback between growth and growth rates which in the context of networks lead to scale-free or single-scale networks. The third, sublinear, regime is characteristic of biological scaling processes and those that that are driven by optimization and efficiency. These systems are characterized by a negative feedback in growth rates and as such naturally exhibit saturations {\Et -- areas where growth ceases \textcolor{black}{from a lack of resources}}. \textcolor{black}{Motivated by this observation, we propose and analyze a simple network growth process that is analogous to this sublinear regime and characterize how its scale-free saturations impact the diversity and fairness of its structural properties and {\Et give rise to } scaling relations observed throughout complex systems and science.}
\end{abstract}

\maketitle

\section{Introduction}
Many natural and man-made growth processes exhibit local saturations, {\Et that is, depletion of resources,} under sustained growth. 
Oftentimes these saturations have detrimental effects on the system, as the competition for resources increases.
The finite number of local resources effectively impose constraints on local growth and eventually forces the global system to expand or change in ways to sustain further growth. 
The rate at which new resources are, or have to be, added to a system with size $N$ largely determines the ``pace of life'' which expresses itself by a power law scaling of the form $N^{\beta-1}$~\cite{bettencourt2007growth}.
Bettencourt and colleagues identify three universal categories that can be briefly outlined as follows.
The linear regime ($\beta=1$) has a constant growth rate resulting in exponential growth associated to the individual level.
The superlinear regime $(\beta>1)$ exhibits increasing growth rates.
This is observed in cities, whose continued growth requires a constant increase in the creation of wealth, resources, and information through innovations, signifying the fast pace of life in large cities.
In contrast, the sublinear regime ($\beta<1$) exhibits decreasing growth rates and thus eventually saturates.
In cities this is found in the infrastructural economies of scale that are focused on \emph{efficiency} and \emph{optimization} rather than innovation.
This is also the regime of biological systems whose pace of life decreases in growth, as exemplified by the allometric quarter power law scalings that predict metabolic rates decrease with body size~\cite{west1997general,brown2002fractal}.\\

We argue that these three general relations between growth, scaling, and resource availability have natural analogies to how growing networks organize themselves.
Preferential attachment, or cumulative advantage, is analogous to superlinear scaling in which there is a positive feedback between growth and growth rates. This famoulsy leads to the rich-get-richer (or Matthew effect) and scale-free networks with power law degree distributions with superlinear exponents that mostly lie between two and three~\cite{newman2005power}.
When preferential attachment is tempered, broad-scale degree distributions emerge that exhibit power-law behaviors with an exponential cut-off~\cite{amaral2000classes,doi:10.1073/pnas.0307625100,doi:10.1073/pnas.0606779104}. 
This exponential decay in ``single scale'' networks is analogous to the linear scaling regime of which uniform attachment is a prime example.
This finally leaves the sublinear regime whose growth is hampered by saturations. 
In networks, saturations occur if sets of vertices cannot receive new connections. 
Depending on how fast the saturations occur this remains a local phenomenon or becomes prevalent throughout the network.
In any case, local saturations tend to increase the \emph{distance} between vertices, thus increasing the diameter and limiting the global connectivity of the network.
This hurts the \emph{efficiency} or \emph{optimality} of information flow and resource exchange over the network, which is are the main driving forces of growth processes in the sublinear regime.  

Thus, it is natural for a network growth process in this regime to avoid local saturations by distributing its edges over the vertices not by how much they already have, but by how much they are \emph{missing}. 
Indeed, this is analogous to the negative feedback between growth and growth rates that naturally arise from the usage of local resources.
These growth processes thus represent the network analogy of how infrastructural economies of scale in cities~\cite{doi:10.1126/science.1235823,bettencourt2007growth,doi:10.1073/pnas.2214254120,5255fb9c-f048-3309-98fd-bd044c76783b,gastner2006optimal} and biological systems~\cite{brown2002fractal,west1997general} tend to grow. 
Yet, they received significantly less scientific attention than the other two scaling regimes and therefore relatively little is known about them.
\textcolor{black}{Statistical analyses however support the idea that real-world \emph{biological} networks may grow in a different scaling regime and have different structural properties than, for example, social and technological networks~\cite{broido2019scale}.
Recent research on \emph{evolving} networks have indeed recognized this, stressing the importance of physicality and the constraints that it imposes on network structure~\cite{posfai2024impact}. 
One of these constraints are saturations, typical to sublinear growth, that can result in networks to not only evolve but also \emph{grow} in a constrained manner.}
The above empirical findings then suggest that scalings may lie hidden in the global structures \textcolor{black}{induced} by local sublinear growth in networks. 
\textcolor{black}{As in scale-free networks, an idealized model can then be used to uncover these scalings from which the basic principles, and mechanisms that underlie the diverse structure of networks can be identified and built upon. }
Furthermore, by \emph{reversing} the network analogy we can also learn about the possible effects that local sublinear growth and scalings may have on the global scale of complex systems \textcolor{black}{across different research domains}.

To investigate this, we focus on two basic but important constraints that naturally induce saturations in growing networks: the minimum vertex degree $d$, and maximum vertex degree $w$.
Both have important implications for networks and the processes and algorithms over them. 
For example, degree bounds commonly influence various spectral properties of networks relevant to dynamical systems that evolve over them~\cite{chung1997spectral},
Degree bounds are also crucial in the feasibility of property testing algorithms~\cite{goldreich1997property} and preference identification over networks~\cite{de2018identifying}. 
In games on networks, the critical contagion threshold is known to be at least $1/w$~\cite[Corollary 3]{06c81afd-70f8-32ca-9705-30ea838a0775,jackson2008social}, indicating that a bound on the maximum degree can aid in the effectiveness of mechanism design. 
Lower bounds on the vertex degree are desirable for the resilience of distributed control~\cite{PIRANI2023111264}, the speed of log linear learning~\cite[Proposition 4]{doi:10.1073/pnas.1100973108},
and the structural (target) controllability of networks~\cite{9273235,liu2011controllability} determined by the size of its maximum matching, often studied in random graphs~\cite{10.1007/978-3-0348-7915-6_11}.

The above illustrates the importance of degree-related properties of networks and motivates their usage as a natural starting point to investigate how they may affect structural and scaling relations in sublinearly growing networks.
 Our contributions in this direction can be summarized as follows:

\begin{itemize}
    \item We propose an intuitive network growth process rooted in empirical sublinear scalings that incorporates a minimum of parameters, specific details, and is complementary to existing models.
    \item Our analysis shows that local power law scalings in the saturation and volume of sets of vertices quickly emerges with a fractional exponent that is easily interpretable as the ratio of resource consumption to production during growth. 
    \item We characterize how the local sublinear scalings express themselves in the emergent global structure and bound finite size effects with an almost sure convergence rate.
    \item We explicitly link the sublinear exponents to the frequently reported Taylor's law or {\Et fluctuation scaling}
    with a constant exponent two and a variable slope that is interpretable by the structural fairness of the network.
    \item We discuss the implications for other important network properties and the possible relevance of the idealized process to other sublinear scalings.
\end{itemize}

\subsection{Outline}
In Section \ref{sec: model} we define the growth process with scale-free local saturations and relate it to sublinear growth processes.
In Section \ref{sec: dist} expressions are derived for the asymptotic structural quantities and convergence rates are provided for individual realizations. We show how the local scale-free saturations express themselves through various properties and scalings at the global scale.
Section~\ref{sec: con} discusses the results and concludes the paper. 
Longer proofs of formal statements are found in the appendices.

\section{The growth process with scale-free saturations}\label{sec: model}

\begin{table*}[]
\begin{tabular}{cccccc}
\textbf{Regime} & \textbf{Driving force} & \textbf{Local principle} & \textbf{Proportionality} & \textbf{Growth rate scaling} & \textbf{Idealized structure} \\ \hline
\rowcolor[HTML]{EFEFEF} 
superlinear            & innovation             & cumulative advantage                        & volume                   & $\lambda^{1/2}$                 & scale-free                   \\
linear                 & individual             & neutral                          & size                     & -                               & geometric                    \\
\rowcolor[HTML]{EFEFEF} 
sublinear              & efficiency             & fairness                         & capacity                 & $\lambda^{-\frac{d}{w-2d}}$     & diverse                     
\end{tabular}
\caption{\textcolor{black}{Overview of the features of the three universal growth and scaling regimes and their analogous idealized network growth processes. The growth rate scaling for the superlinear regime corresponds to linear preferential attachment,} \Et{while for the sublinear regime we state the growth rate in terms of capacity that depends on the minimum degree $d$ and maximum degree $w$. \textcolor{black}{Depending on the ratio of these parameters a variety of structures emerge characterized by distinct  structures and correlations associated to constrained growth.}}}
\label{table: overview}
\end{table*}

We consider a random growing graph process in which a single vertex is added at each discrete time step that subsequently connects to $d$ existing vertices with undirected edges---provided that their degree is no more than some constant $w$. 
We will call these $d$ vertices the \emph{parents} of the new vertex. 
Throughout, we label the vertices in the graph by the time $t$ at which they where added. The age of vertex $t$ is $n-t$, where $n$ is the size of the grown network.
For $n\geq1$, $G^{(d,w)}_{n+1}$ is constructed from $G^{(d,w)}_{n}$ by adding the
new vertex $n+1$ that connects to $d$ vertices in $V_n=\{1,\dots,n\}$ independently with probability
\begin{equation}\label{eq: connection prob}
\mathbb{P}(n+1 \text{\ connects to\ } v\in V_n\mid G^{(d,w)}_n)=\frac{w-d_v(n)}{(w-2d)n+c_0},
\end{equation}
where $d_v(n)$ is the degree of vertex $v$ in $G^{(d,w)}_{n}$ and $c_0\geq 0$ is an offset from the initial graph. 
The normalization factor in the denominator of \eqref{eq: connection prob} is equal to the difference in the \emph{maximum} volume $wn$ of grown part in $G^{(d,w)}_n$ and its volume $2d(n)$ at $n$.
It determines the capacity of the vertices $1,\dots,n$ for new connections with future vertices, and as such we refer to it as the \emph{connection capacity} of $G_n^{d,w}$.
We say a set of vertices $X\subset V_n$ is \emph{saturated} if all vertices in $X$ have the maximum degree, and use $c_X(n)$ to denote the connection capacity of a non-saturated set in $G_n^{(d,w)}$.

The global level of resource consumption and generation suggest three distinct growth regimes.
First, if $w<2d$ the connection capacity of the network decreases in its size and the network cannot grow beyond a certain size. 
Second, if $w=2d$ the connection capacity of the network is constant,
reflecting a situation in which resource consumption equals resource generation.
The third and most interesting case occurs when $w>2d$ for which the connection capacity of the network grows in its size and thus reflects a situation in which resources decrease locally, but grow globally.
To obtain an initial insight for this sustained global growth regime, we define 
\[\nu_X(\lambda,n)=\mathbb{E}\left(\frac{c_X(\lambda n)}{c_X(n)}\mid G^{d,w}_n\right),\quad \lambda\geq 1.\]
and study the behavior of this local expected saturation rate in $n$ and $\lambda.$

\begin{figure*}
    \centering
    \includegraphics[width=\linewidth]{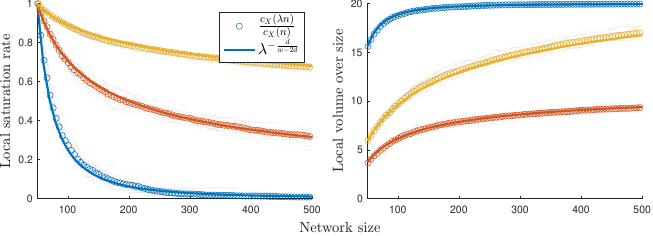}
    \caption{Numerical examples of the power-law saturations in a randomly chosen set in $G_n^{d,w}$ with $|X|=15$ and $n=50$. Different colors correspond to different values of $w$ and $d$ (from top to bottom in left panel): $(w,d)=(40,5)$, $(12,3)$, $(20,8)$.
    The colored plots are averaged over $25$ realizations shown in grey.
   Open circles correspond to the power-law of Proposition~\ref{eq: sat rate}. }
    \label{fig:powerlaw}
\end{figure*}

\begin{proposition}\label{eq: sat rate}
    For $w>2d$, and constant $c_0\geq 0$, the asymptotic expected rate at which sets of vertices saturate in $G_n^{(d,w)}$ is scale invariant under the size of the network. That is, for all $\lambda\geq 1$ and non-saturated $X\in V_n$,
    \[\lim_{n\rightarrow\infty}\nu_X(\lambda, n)=\lambda^{-\frac{d}{w-2d}}.\]
\end{proposition}
\begin{proof}
    The proof follows from an analysis of the rate equation of the expected changes in the connection capacity of sets of vertices, as detailed in Appendix~\ref{app: powerlaw}.
\end{proof}

As in Logan's theory of automatization~\cite{logan1988toward}, the power law in the rate of saturations result from two \emph{counteracting} factors:
the more vertices are added to the network, the more likely it is that the capacity of a particular set decreases from new connections to it, but the more connection capacity is added to the total network by new vertices, the less likely it is that new vertices connect to a particular existing set of vertices. \textcolor{black}{Because these factors act multiplicatively over the dimension of growth, a power-law follows.}

Exponents larger than one represent scenarios in which the ratio of the $d$ resources \emph{used} for growth to the $(w-2d)$ resources \emph{added} during growth is high. 
A natural consequence is that the competition for local resources is high and local saturations occur quick. 
However, most sublinear power law scalings observed ecology, economics, finance, and the infrastructural commodities in cities are characterized by an exponent between zero and one~\cite{brown2002fractal,west1997general,gabaix2009power, bettencourt2007growth}.
The origins of these exponents and their exact values are often a source of scientific research. \textcolor{black}{The allometric laws, for example, are attributed to the fractal branching in resource networks}.
Here, it simply translates to the condition that the maximum degree in a network is no less than three times the minimum degree, i.e., $w\geq 3d$. 
\textcolor{black}{At $w=3d$, the added resources are perfectly balanced with the used resources}.
Past this point $w>3d$, buffers are created which can aid flexibility and robustness to sudden changes.
In Appendix~\ref{app: powerlaw}, we show that in the important sublinear regime with exponents smaller or equal than one, for all $n,\lambda>1$ it holds that
 \[|\nu_X(\lambda, n)-\lambda^{-\frac{d}{w-2d}}|=o(n^{-1+\delta}) \quad \text{for all } \delta>0,\]
and an initial connection capacity that is a non-negative multiple of $w-2d$. 
Markov's inequality then furthermore suggests that large deviations from this power law in the local saturations decay quickly in $\lambda$.

This fast convergence of expectations and ``clustering'' of the exponent is shown in Fig.~\ref{fig:powerlaw} and justifies a discussion on its consequences for \emph{finite} networks and systems.
In particular, we can now confidently \emph{reverse} the network analogy to see that the power law decrease in the rate of the local scale-free saturations is analogous to the decreasing pace of life observed in biological systems and sublinear scaling processes in general.
Furthermore, as a result of the \emph{negative} linear relation between local resource consumption and growth, the local volumes of sets follow an \emph{inverse} power law 
\[\mathrm{vol}_X(\lambda n)=w|X|-\lambda^{-\frac{d}{w-2d}}c_X(n)+o(n^{-1+\delta}),\]
shown in the right panel of Fig.~\ref{fig:powerlaw}. 
\textcolor{black}{This is representative of the eventual saturation of sublinear growth processes and are thus often observed in nature~\cite{jensen1998self,bak2013nature}.}

The connection to continuous growth models can made more explicit from the fact that the leading term of the ``instantenous'' growth rate in the volume of a set is proportional to its \emph{share} in the network's connection capacity, that is,
\[\lim_{\lambda\rightarrow1}\left(\frac{\mathrm{vol}_X(\lambda n)-\mathrm{vol}_X(n)}{(\lambda-1)n}\right)=d\frac{c_X(n)}{c_{V_n}(n)}.\]
This local growth equation for the volume of sets shows that an approximate \emph{dynamic} law of proportional effect holds in which local growth rates \emph{decrease} in the network size {\Et while} global resources increase.

This is akin to the infrastructural economies of scale (for example, number of gasoline stations, length of electrical cables, and road surface) in \cite{bettencourt2007growth} whose number grows in city size but saturate locally. 
\textcolor{black}{In contrast, the growth rate of sets grown by linear preferential attachment \emph{increase} in {\Et $\lambda^{1/2}$} 
representative of the superlinear regime, see Table~\ref{table: overview}}

The negative law of proportional effect has also been reported in various empirical studies on firm sizes~\cite{santarelli2006gibrat}.
In other domains, the dimension of growth of the sublinear scaling may be area, time, frequency, intensity, etc., and the proportionality law may apply to density, value and time, instead of volume.
For example, the carrying capacity of a natural habitat tends to increase in its area, but nested areas within it can saturate~\cite{michael1995species,brown2002fractal,wurtz2008roots}. 
Likewise, growing market demands can increase the number of firms over time, but individual firm stock values and growth rates can saturate~\cite{MULLER19901189,eisler2006scaling}.
Finally, as the intensity or frequency of stimuli increases, local processing, sensory perception, and/or attention may saturate~\cite{marois2005capacity,carandini2012normalization,stevens1970neural}.\\

To understand the extent to which the various sublinear scalings observed in different domains {\Et indeed} can be attributed to (local) saturations and proportional growth, one needs to find appropriate domain specific analogies.
For this {\Et purpose}, the metaphors and parallels such as those between sociological and biological systems~\cite{levine1995organism,bettencourt2007growth}, psychology and neurology~\cite{billock2011honor}, and economics and physics~\cite{STANLEY20011}.
\textcolor{black}{This naturally requires detailed domain-specific models.
 Here, we focus on another important aspect of 
 a simpler, idealized process. It 
 aims to capture how local saturations and sublinear growth results in diverse and fair global structures that exhibit several features found throughout natural and complex systems.}

\section{Emergent global structures}\label{sec: dist}

The exponent $\frac{d}{w-2d}$ of the power law {\Et for }saturations ranges rather broadly from an arbitrarily small positive number up to the minimum degree $d$ in the network. 
We will show that this wide range allows for a diverse set of structural properties to emerge on the global scale that all share the property of being \emph{fair} of which the implications become clear shortly.

{\subsection{Degree distribution}}

Degree distributions are crucial to our understanding of the global structure of networks and the processes that evolve over them.
It is thus natural to ask how they are affected by the local scale-free saturations. 
A good place to start is $w=2d$ \textcolor{black}{at which \emph{any} bounded-degree growth process has a constant connection capacity and}
will therefore result in connected $w$-saturated graphs in which almost all vertices have the maximum degree. The following proposition formalizes this observation for the sublinear networks $G_n^{d,w}$.

\begin{proposition}\label{prop: w=2d dist}
   Suppose the initial graph has connection capacity $c_0>0$,
   then the fraction of saturated vertices in $G^{d,2d}_n$ lies in the interval $[1-c_0/n,1)$. 
\end{proposition}

\begin{proof}
 The proof follows directly from the observation that the number of vertices with a degree less than $w$ is upper bounded by the constant connection capacity~$c_0$.
\end{proof}

Any network grown with $w=2d$ is thus {\Et also} maximally correlated in the sense that every randomly chosen vertex has the same degree as a randomly chosen neighbor asymptotically almost surely.
As seen in Fig.~\ref{fig:powerlaw}, when the maximum degree increases to $w>2d$ saturations quickly slow down and the emergent structure of the growing networks becomes more heterogeneous.
To quantify this for our process, we let $N_k(n)\leq n$ denote the number of vertices with degree $k$ in $G^{(d,w)}_n$.  
The following result establishes convergence of the expected change in the fraction of vertices with degree $k$ in terms of Gamma functions $\Gamma(\cdot)$.
\color{black}
\begin{proposition}\label{thm: dist}
     For $w>2d$ and $k\in\{d,\dots,w\}$ the fraction $\frac{N_k(n)}{n}$ converges to
    \begin{equation}\label{eq: degree dist}
        \rho_k^{(d,w)}=\left(\frac{w}{d}-2\right)\frac{\Gamma(w-d+1)\Gamma(w-k+\frac{w}{d}-2)}{\Gamma(w-k+1)\Gamma(w-d+\frac{w}{d}-1)},
    \end{equation}
  with convergence rate $||N_k(n)/n-\rho_k(w,d) ||=o(n^{-\frac{1}{2}+\delta})$ for all $\delta>0$, almost surely. 
\end{proposition}

\begin{proof}
    The proof is constructed through standard arguments from stochastic approximation methods and can be found in Appendix~\ref{app 1}.
\end{proof}

\noindent The convergence rate of finite networks to the asymptotic expectation allows us to relate the shape and moments of the limit distribution to the local saturations in large but finite networks. {\Et Next, we use this to characterize a notion of fairness, which turns out to be a signifying property, before describing the diversity in degree distributions and correlations in more detail.   }

\begin{remark}
Equation~\eqref{eq: connection prob} implies that, for $w>2d$, the probability that a new vertex chooses the same parent vertex more than once decreases in $O(n^{-2})$. 
Hence, the density of multiple edges and vertices with a degree larger than $w$ approaches zero.
In our calculations, we have therefore lumped these violations into $w$. To completely rule them out, the $d$ edges can be added one and a time.
\end{remark}

\subsection{Fairness, Taylor's law and shape transitions}
The signifying property of the growth process with scale-free saturations is that the edges are distributed among the vertices---not by how much they have---but by how much they are \emph{missing}. 
Intuitively, this leads to a fair and \textcolor{black}{decentralized} allocation that is quite the opposite of the rich-get-richer effect in hub-like scale-free networks.
\textcolor{black}{Both aspects are important and indicative of the contrast between the superlinear and sublinear regime.
For example, infrastructural networks tend to become decentralized from incremental growth as they connect to new parts~\cite{doi:10.1126/science.1235823,5255fb9c-f048-3309-98fd-bd044c76783b,gastner2006optimal}}. 
Moreover, in network allocation problems fairness often plays an integral role in optimality~\cite{kelly1998rate,5461911,gastner2006optimal}.
More recently, it was shown that such an integral approach to fairness is also critical for the accuracy of real-world artificial intelligence systems with inherent hardware constraints in their architecture and circuitry~\cite{guo2024hardware}.
\textcolor{black}{It therefore becomes of interest to understand the structural fairness of networks and how it may within and between different growth regimes.}
Proposition~\ref{thm: dist} allows us to quantify the eventual global fairness of the \emph{structure} of the generated networks by applying Jain's fairness index~\cite{jain1984quantitative,5461911} to the degree distribution~\eqref{eq: degree dist}. 

\begin{proposition}
    For $w>2d$, the asymptotic structural fairness index of $G_n^{(d,w)}$ is
    \[J(d,w)=\frac{4d}{4d+(d+1)(1-\frac{2d}{w})}.\]
\end{proposition}

\begin{proof}
For ease of notation we write $\rho_k^{(d,w)}=\rho_k$. Using \eqref{eq: degree dist}, it is easily verified that the average degree and variance of the asymptotic degree distribution are given by
\begin{equation*}
    \mathbb{E}(\rho_k)=2d, \quad \mathrm{Var}(\rho_k)= d(d+1)\left(1-\frac{2d}{w}\right).
\end{equation*}
The fairness index then follows from its definition $$J(d,w)=\frac{\mathbb{E}(\rho_k)^2}{\mathbb{E}(\rho_k)^2+\mathrm{Var}(\rho_k)}.$$
\end{proof}

We now have a global fairness index in terms of $w$ and $d$ that can be compared across the range of the exponent and other existing network models.
The asymptotically regular graphs achieve the maximum fairness index equal to $1$, while the lowest fairness index $4d/(5d+1)$ is obtained when the maximum degree tends to infinity.
To put this in context: a \emph{static} random graph with mean degree $2d$ has a fairness index of $2d/(2d+1)$, and scale-free networks with an exponent between two and three asymptomatically have the minimum fairness index $0$.
This shows that the generated networks are structurally fair across the range of the exponent and may thus exhibit the beneficial properties associated to fairness.\\

At first sight, it may then seem that the generated networks are all quite similar.
However, the high fairness indices are due to a low coefficient of variation, which normalizes the standard deviation with the mean degree $2d$ and both change similarly across the range of the exponent. 
With a bit of algebra we obtain the following relation.

\begin{corollary}
For all $w>2d$, it holds that 
    \begin{equation}\label{eq: Taylor}
    \mathrm{Var}(\rho_k)=\left(\frac{1}{J(d,w)}-1\right)\mathbb{E}(\rho_k)^2.
\end{equation}
\end{corollary}

\begin{figure}
    \centering
    \includegraphics[width=\linewidth]{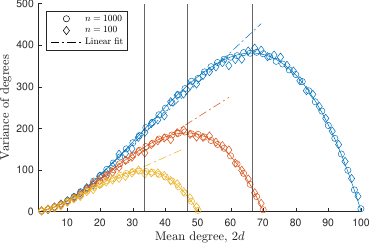}
    \caption{Simulation on the relation between the square of the expected degree and the variance in degree distributions. Different colors correspond to the maximum degrees $w=(50,80,100)$. The values of $d$ range from $1$ to $\lceil w/2\rceil$. 
    The solid curves are computed by~\eqref{eq: Taylor}. 
    The markers show the average variance of $25$ realizations of a network of size $1000$ (circles) and $100$ (diamond). 
    The three vertical lines indicate the points at which the exponent of saturations is equal to one, i.e., $w=3d$m \textcolor{black}{and the degree distribution is uniform}. 
   \textcolor{black}{The dashed dotted curves are linear functions fitted over the exponents of saturations ranging from $1/8$ to $3/4$ for the three $w$ values. Over this range, each fit has a coefficient of determination larger than $0.98$.}}
    \label{fig:Taylor}
\end{figure}
\begin{figure}
    \centering
    \includegraphics[width=\linewidth]{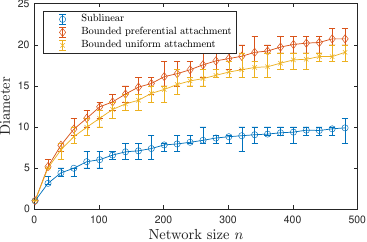}
    \caption{The average diameter of 25 growing networks with a mean degree close to its maximum degree $w=2d+1$. The networks are generated with three different attachment processes: sublinear growth with an exponent of saturation $d$, bounded uniform attachment, and bounded preferential attachment \textcolor{black}{in which new vertices attach uniformly and preferentially to the subset of vertices with a degree less than $w$. }
    The negative feedback in the sublinear growth rates slow down saturations, thereby \textcolor{black}{also} slowing down the increase of the diameter of the network during its growth.} 
    \label{fig:diameter}
\end{figure}
\begin{figure*}
    \centering
    \includegraphics[width=\linewidth]{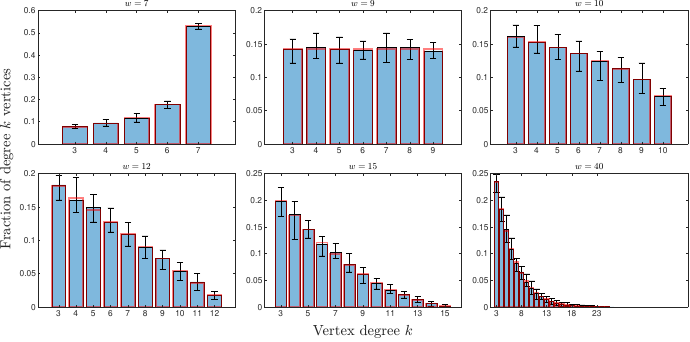}
    \caption{Degree distributions of $G^{d,w}_n$ with $d=3$ and varying $w$. The distributions are averaged over 25 realizations of $G^{(d,w)}_n$ with $n=1001$. The red border corresponds to the analytic degree distribution in~\eqref{eq: degree dist}. 
    Error bars show the maximum and minimum realizations and are in accordance with the $1/\sqrt{n}$ error term of Proposition~\ref{thm: dist}.}
    \label{fig:dists}
\end{figure*}
The mean and variance of degrees can thus be expressed by a simple quadratic relation with a \emph{slope} that is naturally interpreted by the fairness index. 

This relation applies generally to any network with a fairness index larger than zero.
\textcolor{black}{Therefore, the importance of \eqref{eq: Taylor} is not that it exists but rather the \emph{shape} of the variance-mean relation.}  
\textcolor{black}{In particular,
the slowly varying fairness index causes the variance in Fig.~\ref{fig:Taylor} to be left-skewed, implying that, for a broad range of sublinear scaling exponents, the quadratic relation is approximately \emph{linear}.
This is illustrated by the linear slope in Fig.~\ref{fig:Taylor} which has an excellent fit over exponent values ranging from $1/8$ to $3/4$.}
Thus, for the most widely observed sublinear exponents, the local scale-free saturations express themselves at the global scale by degree distributions that approximately follow Taylor's law~\cite{f4eb8472-e94e-3d96-8e2a-f2462ed4df6b,taylor1961aggregation} with a \emph{constant} exponent~$2$.
Remarkably, it is precisely this exponent $2$ that is also widely reported in empirical studies on Taylor's law throughout natural science, information technology, finance, and psychology~\cite{doi:10.1073/pnas.1505882112,wagenmakers2007linear,Eisler01012008}.
\textcolor{black}{We return to this observation in the discussion section}.
\\
The overall bell shape of \eqref{eq: Taylor} is \textcolor{black}{also} indicative of the diversity of the degree distributions whose \emph{skewness} varies significantly across the exponent of saturations.
This is most easily seen by the \emph{shape transitions} in the degree distributions shown in Fig.~\ref{fig:dists}
\textcolor{black}{and the continuous approximation of the degree distributions that suggest $\rho_k\sim (w-k)^{w/d-3}$.}
When the exponent of saturations is larger than one ($2d<w<3d$), the quick saturations result in an increasing degree distribution.
This corresponds to a network that operates close to its maximum capacity and it 
is also here that sublinear growth is particularly important for the network's diameter, see Fig.~\ref{fig:diameter}.
At $w=3d$, the exponent is one and local sets saturate harmonically in the networks size which leads to a \emph{uniform} degree distribution
 \[\rho^{(d,3d)}_k=\frac{1}{1+2d}.\]
Here the variance increases maximally in the mean degree as seen in Fig.~\ref{fig:Taylor}. 
At this point, the volume of the network $2dn$ is two-thirds of its maximum $wn$ and below this density the degree distribution becomes decreasing.
When $w$ increases to $4d$, the exponent is $1/2$ and the degree distribution shows a \emph{linear} decay
\[\rho^{(d,4d)}_k=\frac{2-2k+8d}{2+9d(d+1)}.\]
So while the shape of the distribution changes quite drastically between the exponents $1$ and $1/2$, the fairness index drops no more than $2/35$.
The degree distribution becomes increasingly steep as $w$ grows and approaches the left-truncated \emph{geometric}
\[\lim_{w\rightarrow\infty}\rho^{(d,w)}_k=\frac{1}{1+d}\left(\frac{d}{d+1}\right)^{k-d}
,\]
at which the networks achieve their minimum fairness index and lowest asymptotic density $0$.  
This thus corresponds to networks in which the vast majority of vertices have a large connection capacity. 
In the continuous degree approximation obtained by the rate equation method this limit coincides with an exponential decay in the degree distribution with rate $1-k/d$. 
Indeed, as the exponent of saturations goes to zero the sublinear growth scalings approach linear scalings whose neutral selection principle result in ``single scale'' networks.
In between these cases, the degree distribution exhibits sublinear and superlinear decay, as the sets of vertices saturate with an increasingly small fractional exponent in the size of the network.\\

The importance is that a diverse set of global structures can result from (small) changes in the exponent, while the fairness index of the global structure remains largely unchanged.
Moreover, because the emergent global structures are determined by the fractional exponent they are also \emph{independent} of the scale $c\cdot(w-d)$ with $c$ in $\mathbb{N}^+$.
That is, analogous to how an elephant is a blown up version of a mouse, the degree distribution of a sublinearly growing network with a maximum degree $500$ and minimum degree $30$, is blown up version of a network with a maximum degree $50$ and minimum degree $3$. 
This is however, not the complete picture.  
One may wonder, for example, why do we mostly observe smaller exponents and hetereogenous networks \textcolor{black}{in the real-world} as opposed to the maximally fair asymptotically regular ones?

\begin{figure*}
    \centering
    \includegraphics[width=1\linewidth]{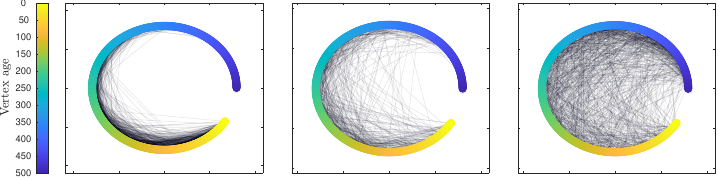} 
    \caption{Generated networks with $500$ vertices consecutively placed on an open circle with increasing vertex age $500-t$. From left to right the values of $(w,d)$ are $(50,24),(7,3),(14,3)$ with corresponding exponents $12,2,3/8$. 
    The network in the left panel is much denser than the other two, but all edges are closely knit between vertices with a similar age, \textcolor{black}{creating age-based communities}. 
    Because sublinear growth avoids local saturations these communities quickly disappear as seen from the increasing number of edges through the interior of the open circle in the middle and right panel. 
    From left to right the network diameters are: $10,10,6$.}
    \label{fig:networks}
\end{figure*}
\begin{figure*}
    \centering
    \includegraphics[width=\linewidth]{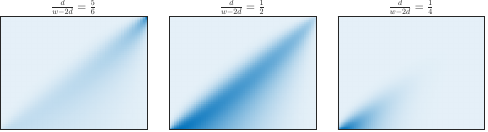}
    \caption{Visuzalization of parent-child degree correlations computed by \eqref{eq: mkl general solution 2}. The lower left corner correspond to $\rho_{d,d+1}$ and the right up corner to $\rho_{w,w}$. For high exponents, the saturations dominate correlations (left panel), for low exponents the growth process dominates (right panel), in between correlations are due to saturations and growth (middle panel). }
    \label{fig:cormatrix}
\end{figure*}

{\subsection{Degree correlations and age-based communities}}
While degree distributions and their shapes provide important clues to the extent of saturated structures in the networks, they fail to capture the underlying correlations that are also often present in the real-world~\cite{PhysRevLett.89.208701}. 
In superlinear network growth models and static models these degree correlations can be attributed to the attachment process and structural cutoffs in finite networks~\cite{refId0,catanzaro2005generation}.
In the presence of local saturations there is one more thing to consider: connected sets of saturated vertices are inherently assortative, and their presence thus has the potential to alter, or even dominate global correlations.
Of course, the extent at which this occurs depends on how fast vertices and their neighbors saturate.
In the sublinear growth process this is thus again closely tied to the exponent of the scale-free saturations.\\

The left panel of Fig.~\ref{fig:networks} shows that when local saturations are fast new connections can only be made to the young tail of the network and age-based communities appear in the structure of the network.
In spatial growth processes, this tail reflects the moving boundary of an expanding area, such as the outer suburbs of a city that have a large capacity for an increase in infrastructural commodities, or the border of an expanding habitat at which there is less competition for resources.
For the maximum exponent values the age-based communities can be so extreme that the diameters of these networks are similar to those that are considerably less dense but have a smaller exponent, as seen in middle panel of Fig.~\ref{fig:networks}.

\begin{figure*}
    \centering
    \includegraphics[width=\linewidth]{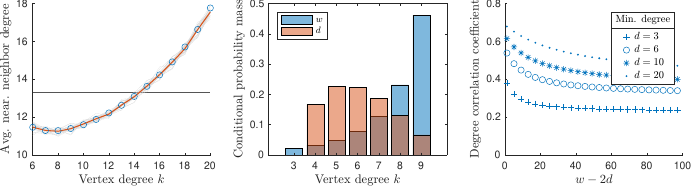}
    \caption{Example of degree correlations in generated networks. The left panel shows an example of the average nearest neighbor degree as function of $k$. Circles correspond to $\sum_{l=d}^w lp(l\mid k)$, the orange line is the average of $25$ realization in a $G_{1001}^{6,20}$ shown in gray. The horizontal line shows the global average. The middle panel shows the conditional degree distribution for vertices with degree $d=3$ and vertices with degree $w=9$ computed using~\eqref{eq: con dist}. The right panel shows the global positive degree correlation coefficient for several $d$ and $w$ computed using the limit joint degree distribution in~\eqref{eq: joint dist}.
    }
    \label{fig:corr}
\end{figure*}
 
When the exponent decreases even further the slowed down saturations allow young vertices to have increasingly older parents, effectively widening the connectivity ``time-window" until, eventually, the age-based communities completely disappear and the diameter becomes shorter and shorter~as seen in the right panel of Fig.~\ref{fig:networks}.
\textcolor{black}{The local saturations thus naturally induce time-windows within connections between vertices are more likely, known as a developmental mechanism for connectomes~\cite{kaiser2017mechanisms,CHIANG20111}. 
In the idealized sublinear process, the rate of saturations and its parametrization into the degree bounds largely determine the length of the window and its impact on the growing network's modularity or community structure. }

\begin{widetext}
        \begin{align}\label{eq: mkl general solution 2}
   \rho_{kl}=\frac{\left({\frac{w}{d}-2}\right){\Gamma(\frac{w}{d}-2+w-k+w-l)\Gamma(w-d+1)^2}}{\Gamma(w-d+\frac{w}{d}-1)\Gamma(w-k+1)\Gamma(w-l+1)}
   \sum_{j=1}^{l-d}\binom{l-d+k-d-j}{k-d}\frac{\Gamma(\frac{w}{d}-1+w-d-j)}{\Gamma(\frac{w}{d}-1+2(w-d)-j)}.
\end{align}
      \end{widetext}

To quantify how this affects the correlations in the network we follow~\cite{PhysRevE.63.066123} and approximate $N_{kl}(n)$, defined as the number of added vertices with degree $k$ that have a parent with degree~$l$. 
Because each added vertex has $d$ links to parents this quantity is related to the degree distribution~\eqref{eq: degree dist} by $\sum_{l}N_{kl}(n)=dN_k(n)$. 
Moreover, since each parent vertex, by construction, has at least degree $d+1$ it must hold that $N_{kd}=0$ for all $k=d,\dots,w$. 
These boundary conditions can be used together with Proposition~\ref{thm: dist} to obtain the following result.\\
\color{black}
\begin{proposition}\label{thm: ref}
For $w>2d$, $k\in\{d,\dots,w\}$ and $l\in\{d+1,\dots,w\}$ the fraction $\frac{N_{kl}(n)}{n}$ converges to $\rho^{(d,w)}_{kl}$ in \eqref{eq: mkl general solution 2}
with convergence rate $||N_{kl}(n)/n-\rho^{(d,w)}_{kl} ||=o(n^{-\frac{1}{2}+\delta})$ for all $\delta>0$ almost surely.
\end{proposition}

\begin{proof}
    The proof follows a similar approach as the proof of Proposition~\ref{thm: dist}. Details can be found in Appendix~\ref{app2}.
\end{proof}

The properties of the derived parent-child degree correlations are clearly visible in Fig~\ref{fig:corr}. 
If saturations are quick, they dominate the global parent-child correlations, which indicates the emergence of age-based communities. 
On the other hand, if saturations are sufficiently slow the correlations are dominated by the sublinear growth process. 
Clearly, these effects can naturally occur in any growth process with (local) saturations. 
What sets the sublinear growth process apart from, for example, bounded preferential attachment or uniform attachment, is that by actively avoiding local saturation---while still utilizing local capacities---mixtures of correlations arise in a \emph{wide} range of the sublinear exponents.

\subsection{Disentangling growth and saturation correlations}

The most interesting structures occur when the degree correlations are affected by both local saturations and growth. 
\textcolor{black}{These cases are best understood by considering the symmetric variable
 $E_{lk}(n)=N_{kl}(n)+N_{lk}(n)$, which captures the number of edges between vertices with degree $k$ and $l$~\cite{fotouhi2013degree}. }
 The convergence of $N_{kl}$ can then be interpreted by more commonly used structural properties~\cite{PhysRevE.81.011102,PhysRevLett.94.188701,PhysRevE.65.056109,PhysRevE.101.022304,WANG2018164,bogua2003epidemic}.
\begin{corollary}
The conditional degree distribution of $G_n^{(d,w)}$ converges to
     \begin{equation}\label{eq: con dist}
    \rho(l\mid k)=\frac{E_{lk}(n)}{kN_k(n)}\rightarrow\frac{\rho_{kl}+\rho_{lk}}{k\rho_k},
    \end{equation}
and the \emph{joint degree distribution} converges to
    \begin{equation}\label{eq: joint dist}
     \rho(l,k)=\frac{E_{lk}(n)}{dn}\rightarrow\frac{\rho_{kl}+\rho_{lk}}{d}.
    \end{equation}
\end{corollary}

\begin{figure*}
    \centering
    \includegraphics[width=\linewidth]{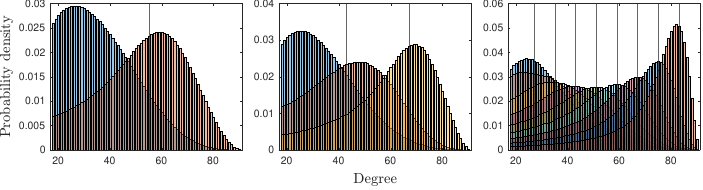}
    \caption{Mixtures of conditional degree distributions for two (left), three (middle), and nine (right) evenly split saturation levels of vertices in a network with $d=18$ and $w=90$. \textcolor{black}{The vertical lines correspond to two (left), three (middle) and eight(right) evenly spaced saturation levels categories.} } 
    \label{fig:mix dist}
\end{figure*}

The right panel of Fig.~\ref{fig:corr} shows how the Pearson correlation coefficient~\cite{PhysRevLett.89.208701,PhysRevE.64.041902} computed by the joint degree distribution decreases, but eventually settles, as the exponent of saturations decreases. 
Together with the \emph{mostly} increasing average nearest neighbor degree plot shown in the left panel this is indicative of the global assortativity of the sublinear networks. 

More importantly, we can now distinguish the structural properties of vertices based on how saturated they. 
The middle panel of Fig.~\ref{fig:corr} provides a simple but striking example that occurs at $w=3d$. 
The ``old'' saturated subgraphs make the conditional degree distribution for vertices with degree $w$ highly left skewed.
In contrast, vertices with degree $d$ have more variance in their connections as they are mostly linked to young vertices in the tail of the network. 
The combination and prevalence of the different degrees make it such that the global degree distribution is uniform.

\textcolor{black}{As in the time-windows mechanism these effects can be lumped into broader, often more practical, categories of vertex saturation levels by mixing the conditional degree distributions proportional to prevalence, see Fig.~\ref{fig:mix dist}}.
\textcolor{black}{The resulting distributions provide information on the likelihood of the degree of a vertex found by following a random edge from a random vertex within a given saturation category. 
The area that overlaps with other categories thus provide a measure for the likelihood of links between different saturation categories, capturing a degree based view of the time-windows and age-based communities shown in Fig.~\ref{fig:networks}.
For broad categories this view is coarse and the effect of full saturations on the global structure is significantly suppressed.
By increasing the number of categories, a more detailed view emerges that highlights the important and gradual structural impact of local saturations on the global structure of sublinearly growing networks.}

\section{Discussion and conclusion}\label{sec: con}
Motivated by a negative feedback mechanisms between local growth rates and growth capacity, we have proposed and analyzed a growing network process that incorporates a minimum of specifics: namely, the minimum and maximum degree, $d$ and $w$, 
\textcolor{black}{and a proportional negative feedback {\Et in new connections}}.
\textcolor{black}{The idealized process models the defining negative feedback of sublinear growth and its saturation by distributing edges among the vertices by how much they are missing, rather than how much they have.}
By avoiding local saturations, but still utilizing local capacities, a dynamic law of proportional effect occurs that quickly leads to power law scaling in the rate at which any local subset of vertices saturates with an exponent $d/(w-2d)$.
As a consequence of the negative linear relation between local capacity and size, 
an \emph{inverse} power law holds in the local volume of sets of vertices.
These inverse power laws also commonly occur in nature~\cite{jensen1998self,bak2013nature} which is indicative of the importance of sublinear growth and saturations. 
{\Et Our work provides mechanistic insights for these in the context of networks. }Our stochastic analysis on the fast convergence of expectations furthermore imply that individual finite networks cluster around the exponent, \textcolor{black}{which is in agreement with simulations.}. 
\textcolor{black}{Despite a remarkable diversity in degree distribution shapes, the process generates structurally fair networks with a low coefficient of variation across the wide range of the exponent. 
A practical implication is that a simple edge allocation process based on connection capacities can aid in maximizing fairness and improve the robustness of network systems, see e.g.~\cite{10.1093/comnet/cnv005} and the references therein.}

From a more theoretical perspective, another implication of the small {\Et range in }
the fairness index is that that for the most commonly observed sublinear exponents (strictly between zero and one), an approximate version of Taylor's law holds. 
This scaling behavior between the mean and variance was originally observed in the spatial distribution of animal populations and later in several other non-ecological systems including neurology, finance, and epidemiology~\cite{f4eb8472-e94e-3d96-8e2a-f2462ed4df6b,taylor1961aggregation,kendal2001stochastic,doi:10.1073/pnas.1505882112,Eisler01012008}.
The widely reported exponents in Taylor's law have also sparked several general explanations and a variety of models are known to (eventually) satisfy Taylor's law~\cite{doi:10.1073/pnas.1505882112,COHEN201430,Eisler01012008}.
Here, we have provided a more explicit link to sublinear growth processes and scale-free saturations, and showed Taylor's law and the exponent two emerges quickly enough for it to be relevant even for small networks.
Numerical simulations furthermore suggest that the approximate Taylor's law also emerges in the aging tails of a network
suggesting similar scaling effects occur at the boundaries of spatially expanding sublinear growth processes such as \textcolor{black}{the infrastructural commodities of growing cities~\cite{doi:10.1073/pnas.1913014117,doi:10.1073/pnas.2214254120}.} 
More analysis is required for a deeper formal understanding of this effect.

\textcolor{black}{In general, the exponent of saturation can be broadly interpreted as a \emph{dimensionless ratio} of resources used and produced during the course of a sublinear growth process. 
The importance is that the scales of $w$ and $d$ offer a naturally interpretable parametrization that can account for reported (species or interpersonal) variability in the slope of Taylor's law~\cite{schmiedek2009relation,brown2002fractal} whose source is {\Et otherwise} often not fully understood~\cite{kello2010scaling}.}
From this perspective our idealized process may be relevant for understanding core principles of much more complex sublinear scalings in other domains.
For example, Taylor's classic study measured the population abundance of a certain species for different habitat area sizes~\cite{Eisler01012008}. 
In habitats with a larger area, local saturations are less important for abundance because of the ``rescue effect''~\cite{keitt1998dynamics,brown1977turnover}. 
Thus, it is not unlikely that larger habitats have smaller exponents than similar smaller ones.
The sublinear growth process then predicts Taylor's law holds if local populations (usually called metapopulations) predominantly grow proportionally to resource availability with exponents in the linear regime of Fig.~\ref{fig:Taylor}.
This, however, does not account for possible decreases in abundance in response to quick local saturations, which would be an important consideration for population dynamics with higher exponents. 
In general, we observe global variability increases in the scale of resources, as illustrated by the higher variance curves of Fig.\ref{fig:Taylor} for larger values of $w$.
Interestingly, regardless of the $w/d$ scale, variability is maximized when the usage and addition of resources are closely aligned ($w\approx 3d$) and on average, vertices utilize two-thirds of their capacity for connections.
\textcolor{black}{Theoretically}, this peak in the connection variability \textcolor{black}{at the critical exponent $1$ is of interest because it shows that resource-efficient and structurally fair networks}, may have an increased dynamic range of pooled responses to stimuli that can vary in several orders of magnitude~\cite{kinouchi2006optimal}. 
\textcolor{black}{Indeed, in this regard the performance of scale-free networks in the superlinear regime have been shown to be lacking for $d>1$~\cite{copelli2007excitable}. 
The sublinear networks provide a simple idealized alternative to further understand critical phenomena on complex networks~\cite{RevModPhys.80.1275}. They offer 
a wide array of complementary properties that range from the highly organized almost $w$-regular growing networks to geometric single scale tails and everything in between.}

To further understand the structure of the sublinearly growing networks, we have shown that at a global level, the assortative nature of the growth process is enhanced by the local saturations in the structure that can arise in any growing network with degree bounds.
\textcolor{black}{As such, saturation processes and the three growth regimes may offer a natural explanation for complex correlation mixtures observed in biological, metabolic, and other type of networks~\cite{hao2011dichotomy,piraveenan2008local}.}
Disentangling the correlations then becomes important because it helps to understand changes in the structural function or role of vertices over time.
This is seen in the moving tails of a network, that are initially a source of new connection capacity, then turn into ``bridges'' between old and new vertices and eventually utilize most of their connection capacity to ensure a high local connectivity and robustness, see Fig.~\ref{fig:mix dist}. 
\textcolor{black}{For large exponents this process is so quick that age-based communities appear in which vertices of a similar age are tightly connected, but old vertices have almost no links to young vertices.}
In these extreme cases, the positive aspects of the maximally assortative~\cite{10.1093/comnet/cnv005}, fair, and almost regular networks come at the cost of an increasing diameter \textcolor{black}{induced by the narrow time windows.}
As the exponent of saturations decreases, however, the age-based communities disappear and more complex mixtures of correlations emerge and time windows broaden. 
\textcolor{black}{We have captured this by mixtures of conditional degree distributions that can be studied at a arbitrary level of detail.}
\textcolor{black}{It is of interest to further quantify how the exponents, and saturations in general, affect the ``natural'' width of time windows in a growing network, and other structural properties such as clustering coefficients that in are no longer homogeneous~\cite{PhysRevE.65.057102}. }
\textcolor{black}{
The observation that the strength of the age-based communities quickly decreases in the exponent of sublinearly growing networks is then significant because it indicates that there is a substantial \textit{diminishing return} in decreasing the diameter of a growing network by increasing its maximum degree $w$, see~Fig.~\ref{fig:networks}. 
In this way, energy or cost deficient buffers of non-utilized connection capacity or resources can be avoided.
For moderately large $w$ values, the surplus of connection capacity makes the diameter rather insensitive to the particular details of attachment \emph{unless} it grows at high capacities. 
This provides further theoretical support for the ubiquity of small diameters of \emph{growing} networks---even in the presence of a \emph{constant} maximum degree bound.} 

We finish the discussion with two rather important questions: \emph{why} would fair sublinear growth processes emerge; and, second, \emph{what} could one expect to see if a network is indeed generated by such a process?
One simple argument for the emergence of a fair and assortative growth process is homophily~\cite{facd5bfb-620b-3f71-8412-0c6d4aa71a2e}. 
The degree of a vertex is the simplest measure of its importance and in growing network the degree of a vertex is an increasing function of vertex age. 
Thus, homophily in these aspects could lead to a tendency of new vertices to connect to vertices with a similar degree. 
Of course, in human-made engineered networks this fair allocation may simply be by design~\cite{https://doi.org/10.1155/2014/612018,guo2024hardware}.  
We argue however, that a more important argument for sublinearly growing networks to emerge is that they \emph{delay} local saturations and prevent competition for resources or other negative effects that the local saturations may have on a global scale---while still utilizing available resources.
The latter is important because if the networks grow only by connection to new vertices, a lot of connection capacity is wasted.
\textcolor{black}{From a discrete choice perspective~\cite{holme2019rare,10.1145/3308558.3313662}, the sublinear growth process optimizes utilities that are logarithmic in the connection capacity of a vertex, which naturally leads to a vertical asymptote towards negative infinity at the maximum degree $w$ that enforces the strict degree bound. 
Within a more bilateral optimization setting~\cite{jackson2003strategic,stanley1971phase}, the evolution of social behavior through (strong) reciprocity may explain a preference to a fair mechanism~\cite{a5d507a6-ffc3-3033-b7e2-6e5ae70bb82b,WILD2023111469}.}
 
\textcolor{black}{
These more detailed considerations highlight the limitations of the sublinear growth process and its idealized structural properties.
A common criticism is that new vertices have to ``know" the connection capacity or degrees of all existing vertices. 
While a natural time window mechanism partially relieves this issue, it is more likely that a proportional feedback results from a more complex underlying process such as those discussed in~\cite{vazquez2003modeling,kaiser2017mechanisms}.
A clear limitation, however, is that of homogeneity. {\Et In our work,} 
$w$ and $d$ are equal for all vertices in the network, which causes strict cutoffs that are unlikely to be observed in real sublinear networks due to natural heterogeinity but also possible sampling biases~\cite{annibale2011you}.
For the latter, the conditional degree distributions like those in Fig.~\ref{fig:mix dist} may be more representative. 
However, inherent differences between vertices in real networks makes it likely that the two parameters $d$ and $w$ are subject to stochasticity. 
For example, they may be sampled from a possibly common and truncated distribution. 
Such considerations are known to produce better fits with empirical data of growth processes~\cite{PhysRevLett.80.1385,keitt1998dynamics} and would be an important consideration in a statistical study on the likelihood of sublinear growth in real networks.}

\textcolor{black}{To conclude, we emphasize that the growth process with scale-free saturations should not be seen as an alternative to other (idealized) processes such as preferential attachment and its many important variants, but rather as a complementary, yet closely related, scaling and growth regime motivated by different driving forces.
Its application lies predominantly in the ``body'' of networks with finite support that have received considerably less scientific attention, but nevertheless, capture essential features of large scale complex systems whose growth are constrained by local saturations.}

\begin{acknowledgments}
This work was partially funded by Wallenberg AI, Autonomous Systems and Software Program (WASP) funded by the Knut and Alice Wallenberg Foundation and the Swedish Research Council through Grant 2019-00691, the Swedish Research Council under the
grant 2021-06316, and by the Swedish Foundation for
Strategic Research.
The authors would like to thank S\'ergio Pequito and Fiona Skerman for their helpful suggestions and discussions.
\end{acknowledgments}

\section{Appendices}
\appendix
{\section{Proof of Proposition~1 and finite size bounds}\label{app: powerlaw}
Let $c_0$ denote the initial connection capacity. Using \eqref{eq: connection prob}, the expected changes in the connection capacity of $X\subset V_n$ after the addition of a new vertex to the network is 
\[\mathbb{E}(c_{n+1}(X)\mid G_n^{d,w})=c_n(X)-d\frac{c_n(X)}{c_n(V_n)},\]
with $c_n(V_n)=(w-2d)(n+1)+c_0$.
By extending this the expected connection capacity of $X$ after the addition of $j$ vertices to $G_n^{d,w}$ can be recursively written as
\begin{multline*}
    \mathbb{E}(c_{n+j}(X)\mid G_n^{d,w})=c_n(X)\prod_{i=0}^{j-1}\left(1-\frac{d}{c_{n+i}(V_n)}\right).
\end{multline*}
The product on the RHS determines the expected rate at which $c_n(X)$ decreases.
By growing the network from size $n$ to $\lambda n$, $j=(\lambda-1)n$ vertices are added and the expected rate becomes
\[\nu_X(\lambda,n)=\frac{\Gamma(n+\frac{c_0}{w-2d})\Gamma(\lambda n-\frac{d}{w-2d}+\frac{c_0}{w-2d})}{\Gamma(n-\frac{d}{w-2d}+\frac{c_0}{w-2d})\Gamma(\lambda n+\frac{c_0}{w-2d})}.\]
The well-known asymtotic relation $\Gamma(x+\alpha)\approx\Gamma(x)\alpha^x$ then implies that the limit of $\nu_X(\lambda,n)$ as $n\rightarrow\infty$ is independent of the constant initial $c_0\geq 0$ connection capacity and given by the power law $\lambda^{-\frac{d}{w-2d}}$}.

To get an understanding of the behavior for finite $n$, we apply Gautshi's inequality which states that for a positive real $x$ and $0<s<1$
\[x^{1-s}<\frac{\Gamma\left(x+1\right)}{\Gamma(x+s)}<(x+1)^{1-s}.\]
However, because our shift in negative and can be larger than one, we first have to do some manipulation.
For ease of exposition we assume $c_0=0$, but the below arguments can be used for any $c_0$ that is non-negative multiple of $(w-2d)$ by using a positive integer shift in $n$.
Now, observe that 
\[0<\frac{d}{w-2d}<d\Rightarrow f=\left\lceil{\frac{d}{w-2d}}\right\rceil\in\{1,\dots, d\}\]
Then, we can set $s$ to be the fractional part, 
\[s=\left\lceil{\frac{d}{w-2d}}\right\rceil-\frac{d}{w-2d},\]
and set $x$ from the relation
\[x+s=n-\frac{d}{w-2d}\Rightarrow x=n-\left\lceil{\frac{d}{w-2d}}\right\rceil.\]
Now for $n>f$, $x=n-f$ is a positive integer and so we can repeatedly shift the argument of $\Gamma(x)$ up to $\Gamma(n)$ to obtain
\begin{align*}
    \Gamma(n)&=\Gamma\left(x\right)\left(n-f\right)(n-f+1)\dots (n-1)
\end{align*}
And because $\Gamma(x+1)=x\Gamma(x)$, we also have
\[\Gamma(x+1)=\frac{\Gamma(n)}{\left(n-f+1\right)\dots (n-1)}=\frac{\Gamma(n)}{\left(n-f+1\right)^{{\overline{f-1}}}},\]
where $x^{\overline{n}}$ denotes the rising factorial.
And thus
\[\frac{\left(n-f+1\right)^{{\overline{f-1}}}}{(n-f)^{f-1-\frac{d}{w-2d}}}<\frac{\Gamma(n)}{\Gamma(n-\frac{d}{w-2d})}<\frac{\left(n-f+1\right)^{{\overline{f-1}}}}{(n-f+1)^{f-1-\frac{d}{w-2d}}}\]
Now, by setting $\tilde{x}=\lambda n$ and keeping $\tilde{s}=s$, Gautshi's inequality also implies
\[\frac{\left(\lambda n-f+1\right)^{{\overline{f-1}}}}{(\lambda n-f)^{f-1-\frac{d}{w-2d}}}<\frac{\Gamma(\lambda n)}{\Gamma(\lambda n-\frac{d}{w-2d})}<\frac{\left(\lambda n-f+1\right)^{{\overline{f-1}}}}{(\lambda n-f+1)^{f-1-\frac{d}{w-2d}}}.\]
By combining the bounds we obtain an overall bound on the rate
\[\frac{\left(\frac{\left(n-f+1\right)^{{\overline{f-1}}}}{(n-f)^{f-1-\frac{d}{w-2d}}}\right)}{\left(\frac{\left(\lambda n-f+1\right)^{{\overline{f-1}}}}{(\lambda n-f+1)^{f-1-\frac{d}{w-2d}}}\right)}<\nu_X(n,\lambda)<\frac{\left(\frac{\left(n-f+1\right)^{{\overline{f-1}}}}{(n-f+1)^{f-1-\frac{d}{w-2d}}}\right)}{\left(\frac{\left(\lambda n-f+1\right)^{{\overline{f-1}}}}{(\lambda n-f)^{f-1-\frac{d}{w-2d}}}\right)}\]
We can now readily observe that since $1\leq f\leq d$, for any fixed $w>2d>0$ independent of $n$, the limits of the upper and lower bounds as $n\rightarrow\infty$ coincide and evaluate as the power law $\lambda^{-d/(w-2d)}$.
By the squeeze theorem, this must also be the limit of the rate given by the fraction of gamma functions, which is consistent with the earlier derivation.
Any $c_0=c(w-2d)$ for constant $c$, just shifts $n$ to $n+c$, which implies the bounds become increasingly tight. For general $c_0>0$, the shifting property cannot be used, and other techniques need to be used to obtain bounds. 
We now consider the important case $w\geq 3d$ for which $\frac{d}{w-2d}\leq 1$ and thus $f=1$. 
This simplifies the bounds to
\begin{equation}
    \frac{(n-1)^{\frac{d}{w-2d}}}{(\lambda n)^{\frac{d}{w-2d}}}<\nu_X(n,\lambda)<\frac{n^{\frac{d}{w-2d}}}{(\lambda n-1)^{\frac{d}{w-2d}}}.
\end{equation}
Now, observe that the lower bound is increasing in $n$ and the upper bound is decreasing in $n$. Thus, the difference between the power law limit is smaller than the maximum absolute difference of the lower and upper bound indicating the fast convergence of the mean to the power law.

We can use the bounds on the expected rate to bound the probability that deviations occur from the expection. Let $Y_t$ denote the random variable reflecting the connection capacity of a set $X\in V_n$ $t\geq n$. Clearly, $Y_t$ is non-negative for all $t$ almost surely. Thus we can apply the conditional Markov's inequality:
\[\mathrm{Pr}(Y_{\lambda n}\geq a \mid Y_n)\leq \frac{\mathbb{E}(Y_{\lambda n}\mid Y_n)}{a}.\]
So that multiplicative deviations \emph{above} the mean with $\alpha\in(0,1)$ satisfy
\[\mathrm{Pr}(Y_{\lambda n}\geq \lambda^{-\frac{d\alpha}{w-2d}}Y_n \mid Y_n)\leq \frac{(\lambda-\frac{1}{n})^{-\frac{d}{w-2d}}}{ \lambda^{-\frac{d\alpha}{w-2d}}}.\]
Furthermore, because $Y_t$ is non-increasing we can also apply the variation of Markov's inequality that states
\[\mathrm{Pr}(Y_{\lambda n}\leq b \mid Y_n)\leq \frac{Y_n-\mathbb{E}(Y_{\lambda n}\mid Y_n)}{Y_n-b}.\]
So that multiplicative deviations \emph{below} the mean for $\beta>1$ satisfy
\begin{align*}
\mathrm{Pr}(Y_{\lambda n}\leq \lambda^{-\frac{-\beta d}{w-2d}} Y_n \mid Y_n)
\leq \frac{1-(\frac{n-1}{n})^{\frac{d}{w-2d}}\lambda^{-\frac{d}{w-2d}}}{1-\lambda^{-\frac{\beta d}{w-2d}}}.
\end{align*}
In the same the way linear shifts from the exponent can be bounded.

\section{Proof of Proposition~3}\label{app 1}
 The proof is based on an application of Proposition 3.1.1 in \cite{chen2005stochastic}. 
   We first derive the mean path that the random variable $N_k(n)$ would follow if in each step the actual change is equal to the expected change. 
   Using \eqref{eq: connection prob}, for $k\in\{d+1,w\}$ we have
\begin{align}\label{eq: mb k>d}
    &\mathbb{E}(N_k(n+1)-N_k(n)\mid G_n)\nonumber\\
    &=\frac{w-(k-1)}{(\frac{w}{d}-2)n+c_0}N_{k-1}(n)-\frac{w-k}{(\frac{w}{d}-2)n+c_0}N_k(n)\nonumber\\ 
    &=\frac{w-(k-1)}{(\frac{w}{d}-2)}\frac{N_{k-1}(n)}{n}-\frac{w-k}{(\frac{w}{d}-2)}\frac{N_k(n)}{n}+o(n^{-1}).
\end{align}
The positive (negative) factor on the RHS corresponds to the expected increase (decrease) due to new connections with vertices of degree $k-1$ ($k$). 
The term $o(n^{-1})$ comes from the vanishing influence of the initial connection capacity $c_0$. 
Similarly, for $k=d$ we have
\begin{align}\label{eq: mb k=d}
    \mathbb{E}(N_d(n+1)-&N_d(n)\mid G_n)=1-\frac{w-d}{(\frac{w}{d}-2)n+c_0}N_d(n)\nonumber\\
    &1-\frac{w-d}{(\frac{w}{d}-2)}\frac{N_d(n)}{n}+o(n^{-1}).
\end{align}
where the $1$ corresponds to the new vertex and the negative part corresponds to the expected decrease due to new connections with vertices of degree $d$. 
These equations are also known as the rate equations that are extensively used in the network science literature.
We continue to apply the stochastic approximation method to verify the random process does not stray too far from this mean path.
To this end, let $X_k(n):=N_k(n)/n$. Using ${N_k}\left(\frac{1}{n+1}-\frac{1}{n}\right)=-\frac{N_k}{n(n+1)}$ we have
\begin{multline*}
    \mathbb{E}(X_k(n+1)-X_k(n)\mid G_n)\\
    =\frac{1}{n+1}\left(\mathbb{E}(N_k(n+1)-N_k(n)\mid G_n)-X_k\right).
\end{multline*}
Define the functions 
\begin{equation}\label{eq: func roots}
\begin{split}
    f_d(x_d,\dots,x_w)&=1-\left(\frac{w-d}{\frac{w}{d}-2}+1\right)x_d\\
    f_k(x_d,\dots,x_w)&=\frac{w-(k-1)}{\frac{w}{d}-2}x_{k-1}-\left(\frac{w-k}{\frac{w}{d}-2}+1\right)x_k.
\end{split}
\end{equation}
We may then have that for all $k\in\{d,w\}$,
 \begin{multline*}
      \mathbb{E}(X_k(n+1)-X_k(n)\mid G_n)\\=\frac{1}{n+1}\left(f_k(X_d(n),\dots X_w(n))+{o(n^{-1})}\right).
 \end{multline*} 
Define the vector $Z(n):=\left(X_d(n)\dots X_w(n)\right)$ and stack the above functions into the vector $F(x_d,\dots,x_w):=(f_d(x_d,\dots,x_w),\dots, f_w(x_d,\dots,x_w))$. 
By linearity of expectation and the fact that $\mathbb{E}(Z(n)\mid G_n)=Z(n)$, 
we may write 
\[Z(n+1)-Z(n)=\frac{1}{n+1}\left(F(Z(n))+E_{n+1}+R_{n+1}\right).\]
where, as required, {$R_{n+1}=o(n^{-1})$ reflects the vanishing effect of initial vertex, and
$$E_{n+1}=(n+1)\left(Z(n+1)-\mathbb{E}(Z(n+1)\mid G_n)\right).$$ 
Proposition~3.1.1 of~\cite{chen2005stochastic} then tells us that, provided some conditions $A.3.1.1-A.3.1.4$ on $F$ and $E_n$ are met, $Z(n)$ will converge to the root of $F$ in Euclidean norm with rate $o(n^{-\delta})$ for some $\delta>0$.
Before we check the conditions we find the roots of $F$. 
From \eqref{eq: func roots} we find
\[\rho_d=\frac{1}{1+\frac{w-d}{\frac{w}{d}-2}}\]
is the unique root of the function $f_d$.
For $x_k$ with $k\in\{d+1,\dots,w\}$ the roots are given by the solution to the recursion
\[x_k=\left(\frac{w-(k-1)}{w-k+\frac{w}{d}-2}\right)x_{k-1}.\]
Using $x_d=\rho_d$, we find the root of $f_k$ is given by
\[\rho_k=\left(\frac{w}{d}-2\right)\frac{\Gamma(w-d+1)\Gamma(w-k+\frac{w}{d}-2)}{\Gamma(w-k+1)\Gamma(w-d+\frac{w}{d}-1)}.\]
We now proceed to check the conditions.
Condition A.3.1.1 is a commonly accepted requirement for the decreasing step size $\frac{1}{n}$ and are easy to verify.
First, it holds that the limit $n\rightarrow \infty$ is zero and its sum is infinite. 
We also have that
\[\frac{\frac{1}{n}-\frac{1}{n+1}}{\frac{1}{n(n+1)}}=1>0,\]
as desired.
Next since the functions $f_k$ are linear they are totally differentiable and \emph{locally} bounded. Moreover, since it is a function of the expectation of random variables it is also measurable. 
Conditions A3.1.2 and A3.1.4 pertain to the stability and boundedness of the mean path and can be verified using the derivative matrix (or Jacobian) $J$ of $F(x_d,\dots,x_w)$ that contains the partial derivatives $\frac{\partial f_k}{\partial x_j}$ for $k,j\in\{d,\dots,w\}$. 
Due to the structure of $F$, $J$ is lower bidiagonal with a strictly positive lower diagonal
\[\left[\frac{w-d}{\frac{w}{d}-2},\frac{w-d-1}{\frac{w}{d}-2},\dots,\frac{2}{\frac{w}{d}-2},\frac{1}{\frac{w}{d}-2}\right],\]
and \emph{strictly negative} diagonal
\begin{multline}
    \Bigg[-\left(\frac{w-d}{\frac{w}{d}-2}+1\right),-\left(\frac{w-d-1}{\frac{w}{d}-2}+1\right),\dots,\\-\left(\frac{1}{\frac{w}{d}-2}+1\right),-1\Bigg].
\end{multline}
Since the derivate matrix is triangular its eigenvalues
are equal to its diagonal elements. Since these are strictly negative, the derivative matrix is stable for all $w>2d$.
Moreover, since the largest eigenvalue is equal to $-1$, 
the matrix $J+\delta I$ is also stable for all $\delta\in[0,1)$.
This implies both $A3.1.2$ and $A3.1.4$ hold.
We now look at the condition $A.3.1.3$ on the ``noise term'' $E_{n}$. First note that the expectation of the noise term $\mathbb{E}(E_{n+1}|G_n)$ evaluates as
\begin{align*}
    (n+1)\left(\mathbb{E}(Z(n+1)\mid G_n)-\mathbb{E}(Z(n+1)\mid G_n)\right)=0.
\end{align*}
Hence, $E_n$ is a martingale difference sequence with respect to $G_n$.
Moreover, since
\begin{align*}
    ||E_{n+1}||&=(n+1)\left(||Z(n+1)-\mathbb{E}(Z(n+1)\mid G_n)||\right)\\
    &\leq 4(d+1)(w-d+1).
\end{align*}
it is satisfied that
\[\sup_n\mathbb{E}\left(||E_{n+1}||^2\mid G_n\right)<\infty. \]
Then, by the convergence Proposition for martingale difference sequences we have
\[\sum_{n=1}^\infty\frac{1}{n^{(1-\delta)}}E_{n+1}<\infty,\]
almost surely for all $\delta<\frac{1}{2}$. This implies A3.2.2 holds almost surely for all $\delta<\frac{1}{2}$ and completes the proof.

\section{Proof of Proposition~5}\label{app2}

The proof follows the same approach as the proof of Proposition~\ref{thm: dist}. 
That is, we start by deriving the difference equation for the expected change of $N_{kl}$. 
For simplicity we suppress the vanishing effect of the initial vertex that can be handled in the same manner as in the proof of Proposition~\ref{thm: dist} and can be combined with other noise terms that arise from the degree distribution.
We proceed to find the stochastic approximations process for $N_{kl}/n$ and verify the conditions $A.3.1.1-A.3.1.4$ of Proposition 3.1.1 in \cite{chen2005stochastic}. 
The main differences are due to additional noise terms that arise from the degree distribution, and the slightly more complex mean behavior of $N_{kl}$ compared to $N_{k}$. 
We set $N_{0l}=0$. 
Then the expected change in $N_{kl}(n)$ follows the difference equation
\begin{multline}\label{eq: re as}
    \mathbb{E}(N_{kl}(n+1)-N_{kl}(n)\mid G_n)=\frac{w-(k-1)}{(\frac{w}{d}-2)n}N_{k-1l}(n)\\-\frac{w-k}{(\frac{w}{d}-2)n}N_{k,l}(n)+\frac{w-(l-1)}{(\frac{w}{d}-2)n}N_{kl-1}(n)\\-\frac{w-l}{(\frac{w}{d}-2)n}N_{k,l}(n)+{{\frac{w-(l-1)}{(\frac{w}{d}-2)n}N_{l-1}(n)\delta_{kd}}},
\end{multline}
where $\delta_{kd}=1$ if $k=d$ and zero otherwise.
The first (second) term corresponds to the increase (decrease) of an additional connection of a vertex with degree $k-1$ ($k$) already connected to an earlier vertex with degree $l$. The third and fourth terms are due to the same process with $k$ and $l$ switched. The final term is an increase due to the new vertex with degree $d$.

By Proposition~1 we know that $N_{k}/n\rightarrow\rho_k$ with rate $o(n^{-\delta})$ for all $\delta<\frac{1}{2}$ almost surely. 
We will use this to bound the influence of the noise term in \eqref{eq: re as} when verifying condition $A.3.1.3$.
As before, we may write
\begin{multline}\label{eq: expt}
\mathbb{E}(X_{kl}(n+1)-X_{kl}(n)\mid G_n)\\=\frac{1}{n+1}\left(\mathbb{E}(N_{kl}(n+1)-N_{kl}(n)\mid G_n)-X_{kl}(n).\right)
\end{multline}
Define the column vector of placeholder variables $v_k=(x_{kd+1},\dots,x_{kw})$. Let $x_l=\rho_{l}$.
Note that $v_k$ contains $w-d$ elements.
Next, stack the vectors $v_k$ into a single vector $x=(v_d,v_{d+1},\dots,v_w)$, now containing $(w-d+1)(w-d)$ elements.
We proceed to define functions that describe the ``mean path'' analogous to the functions $f_k$ in the proof of Proposition~\ref{thm: dist}.
For $k\in \{d+1,\dots,w\}$ and $l\in\{d+1,\dots,w\}$ let
\begin{multline}\label{eq: system1}
    h_{kl}(x)=\frac{w-(k-1)}{\frac{w}{d}-2}x_{k-1l}+\frac{w-(l-1)}{\frac{w}{d}-2}x_{kl-1}\\-\left(\frac{2w-k-l}{\frac{w}{d}-2}+1\right)x_{kl}+{\frac{w-(l-1)}{\frac{w}{d}-2}x_l\delta_{kd}}.
\end{multline}
For $k\in\{d+1,\dots,w\}$ and  $l=d+1$
\begin{multline}
    h_{kl}(x)=\frac{w-(k-1)}{\frac{w}{d}-2}x_{k-1,d+1}\\-\left(\frac{2w-k-(d+1)}{\frac{w}{d}-2}+1\right)x_{kd+1}.
\end{multline}
For $k=d$ and  $l\in\{d+2,\dots,w\}$
\begin{multline}
    h_{dl}(x)=\frac{w-(l-1)}{\frac{w}{d}-2}x_{d,l-1}+{\frac{w-(l-1)}{\frac{w}{d}-2}\rho_{l-1}}\\-\left(\frac{2w-d-l}{\frac{w}{d}-2}+1\right)x_{dl}.
\end{multline}
Finally for $k=d$, $l=d+1$
\begin{equation}
    h_{dd+1}(x)={\frac{w-d}{\frac{w}{d}-2}\rho_{d}}-\left(\frac{2w-2d-1}{\frac{w}{d}-2}+1\right)x_{dd+1}.
\end{equation}

Let $W(n)$ be the vector of random variables $X_{kl}(n)$ with the same labeling as placeholder $x$.
Using \eqref{eq: re as}, \eqref{eq: expt}, and the above equations we have
\begin{multline*}
    \mathbb{E}(X_{kl}(n+1)-X_{kl}(n)\mid G_n)\\
    =\frac{1}{n+1}\left(h_{kl}(W(n))+\{\frac{w-(l-1)}{\frac{w}{d}-2}(\rho_{l-1}-X_{l-1}(n))\delta_{kd}\right).
\end{multline*}
We already know that for $w>2d$, for all $l=d+1,\dots,w$, the error term $\rho_{l-1}-X_{l-1}(n)$ goes to zero with rate $o(n^{-\delta})$ for all $\delta<1/2$ almost surely as desired for condition $A.3.1.3$, which can also capture the vanishing effect of the initial connection capacity that drops in $o(n^{-1+\alpha})$ for all $\alpha>0$}.
Now, we stack the $(w-d+1)(w-d)$ functions $h_{kl}$ into a single vector $H$, so that we may write
\begin{equation}
   \mathbb{E}(W(n+1)-W(n)\mid G_n)=\frac{1}{n+1}\left(H(W(n))+R_{n+1}\right),
\end{equation}
{where $R_{n+1}$ is a $(w-d+1)(w-d)$ vector that is zero everywhere expect for the first $w-d$ terms that are
equal to the error terms from the degree distribution.} 
Again, by linearity of expectation we have
\[W(n+1)-W(n)=\frac{1}{n+1}\left(H(W(n))+\bar{E}_{n+1}+R_{n+1}.\right)\]
where 
\[\bar{E}_{n+1}=(n+1)\left(W(n+1)-\mathbb{E}(W(n+1)\mid G_n)\right).\]
Notice that $\bar{E}_{n+1}$ is a martingale difference sequence with respect to $G_n$ since its expected value given $G_n$ is zero.
Moreover, as $E_n$ before, the changes in $\bar{E}_n$ are bounded as required.
We now check the derivative matrix of $H$ that has the block matrix structure
\begin{equation*}
A=\left[ 
  \begin{array}{c c c c c c} 
     D_{d} & \mathbf{0} & \mathbf{0}&\dots &\mathbf{0}\\ 
     C_{d+1} & D_{d+1} & \mathbf{0} &\hdots&\mathbf{0}\\ 
     \mathbf{0} & C_{d+2} & D_{d+2}&\hdots &\mathbf{0}\\
     \vdots & \ddots& \ddots&\ddots&\vdots\\
     \mathbf{0}&\hdots & \mathbf{0}& C_{w}&D_{w}
  \end{array} 
\right] 
\end{equation*}
Each block $D_k$ is a $(w-d)$ square matrix. The subscript corresponds to the degree of a vertex, with $k$ in $\{d,\dots,w\}$. 
Each $D_k$ is a bidiagonal matrix with a strictly negative diagonal given by
\[-\left(\frac{2w-k-l}{\frac{w}{d}-2}+1\right), \quad l=d+1,\dots,w.\]
The lower diagonal terms of $D_k$ are given by
\[\frac{w-l+1}{\frac{w}{d}-2},\quad l=d+2,\dots,w.\]
The $C_k$ blocks on the lower diagonal are scaled identity matrices given by
\[C_k=\left(\frac{w-k+1}{\frac{w}{d}-2}\right)\mathbf{I}_{w-d}.\]
The eigenvalues of $A$ are given by the eigenvalues of the diagonal blocks, that by themselves are bidiagonal with a strictly negative diagonal. Consequently, the eigenvalues of $A$ are also strictly negative with the largest eigenvalues equal to $-1$.

Finally, we show the systems of equations in $H$ has a unique root which can be found by solving the recursion
\begin{multline}\label{eq: rec}
    x_{kl}=\left(\frac{w-k+1}{\frac{w}{d}-2+w-k+w-l}\right)x_{k-1l}\\+\left(\frac{w-l+1}{\frac{w}{d}-2+w-k+w-l}\right)x_{kl-1}\\
    +\frac{w-l+1}{\frac{w}{d}-2+w-k+w-l}\rho_{l-1}\delta_{kd}.
\end{multline}
with the boundary conditions
\begin{equation}\label{eq: bc}
\begin{split}
    x_{kd}&=0 \quad \forall k\in\{d,\dots, w\};\\
    x_{k-1d}&=0 \quad k=d;\\
    x_{kl-1}&=0 \quad l\in\{d,d+1\};\\
    \rho_{l-1}&=0 \quad l=d.
    \end{split}
\end{equation}
Note that $\rho(l-1)$ is the probability mass at $l-1$ of the degree distribution given by Proposition~\ref{thm: dist}.
By multiplying the recurrence \eqref{eq: rec} on both sides with
\[f(k,l)=\frac{\Gamma(w-k+1)\Gamma(w-l+1)}{\Gamma(\frac{w}{d}-2+w-k+w-l)},\]
we simplify it to the recurrence relation with constant coefficients $m_{kl}=f(k,l)x_{kl}$ given by 
\[m_{kl}=m_{k-1l}+m_{kl-1}+\eta c(l)\delta_{kd},\]
{where $\eta$ and $c(l)$ are}
\begin{itemize}
    \item $c(l)=\frac{\Gamma(\frac{w}{d}-1+w-l)}{\Gamma(\frac{w}{d}-1+w-d+w-l)}$,
    \item $\eta=\left({\frac{w}{d}-2}\right)\frac{\Gamma(w-d+1)^2}{\Gamma(w-d+\frac{w}{d}-1)}$.
\end{itemize}
This recursion can be solved with back substitution using the boundary conditions~\eqref{eq: bc} to find the unique solution
\begin{align}
    m_{kl}
    =\eta\sum_{j=1}^{l-d}\binom{l-d+k-d-j}{k-d}c(d+j).\label{eq: mkl general solution}
\end{align}
The solution to the original recurrence relation~\eqref{eq: rec} in the statement is then found by reverting the transformation $x_{kl}=m_{kl}/f(k,l)$.

\bibliography{bib}

\end{document}